\documentclass[11pt,twoside,noinfoline]{imsart}

\usepackage{amsfonts,psfrag,epsfig}
\usepackage{color}
\usepackage{amsmath,amssymb}
\usepackage[amsmath,thmmarks]{ntheorem}
\newtheorem{theorem}{Theorem}
\newtheorem{lemma}[theorem]{Lemma}

\newtheorem{claim}[theorem]{Claim}

\newtheorem{definition}{Definition}

{
\theoremstyle{nonumberplain}
\theoremsymbol{\ensuremath{\Box}}
\theoremheaderfont{\normalfont\itshape}
\theorembodyfont{\normalfont}
\newtheorem{proof}{Proof.}

\theoremstyle{empty}

}

\newcommand{\beq}{\begin{eqnarray}}
\newcommand{\eeq}{\end{eqnarray}}
\newcommand{\beqn}{\begin{equation}}
\newcommand{\eeqn}{\end{equation}}
\newcommand{\R}{\mathbb{R}}
\newcommand{\N}{\mathbb{N}}
\newcommand{\Rp}{\mathbb{R}_+}
\newcommand{\ep}{\varepsilon}
\newcommand{\beps}{\varepsilon}
\newcommand{\lf}{\left}
\newcommand{\rf}{\right}

\newcommand{\bd}{\mathbf{d}}

\newcommand{\bone}{\mathbf{1}}

\newcommand{\mc}{\mathcal}
\newcommand{\mb}{\mathbf}

\renewcommand{\hat}{\widehat}
\newcommand{\cM}{\mc{M}}
\newcommand{\cX}{\mc{X}}

\newcommand{\cH}{\mc{H}}

\renewcommand{\P}{\mc{P}}

\newcommand{\bB}{\mathbf{B}}
\newcommand{\bx}{\mb{x}}

\newcommand{\xave}{x_{\text{ave}}}

\begin{document}
\begin{frontmatter}
\title{Distributed averaging via lifted Markov chains}
\runtitle{Fast averaging algorithms}

\begin{aug}
  \author{\fnms{Kyomin} \snm{Jung} \qquad  \fnms{Devavrat} \snm{Shah} \qquad \fnms{Jinwoo} \snm{Shin}\thanksref{shin}\ead[label=e1]{jinwoos@mit.edu}}
  \runauthor{Jung-Shah-Shin}
  \affiliation{MIT}
\thankstext{shin}{Author names appear in the alphabetical order of their
last names. All authors are with Laboratory for Information and
Decision Systems, MIT. This work was supported in parts by NSF projects
HSD 0729361, CNS 0546590, TF 0728554 and DARPA ITMANET project.
Authors' email addresses: {\tt \{kmjung, devavrat, jinwoos\}@mit.edu}}
\end{aug}

\begin{abstract}
Motivated by applications of distributed linear estimation, distributed control
and distributed optimization, we consider the question of designing linear iterative
algorithms for computing the average of numbers in a network. Specifically, our interest
is in designing such an algorithm with the fastest rate of convergence given
the topological constraints of the network. As the main result of this paper,
we design an algorithm with the fastest possible rate of convergence using a
non-reversible Markov chain on the given network graph.  We construct such a
Markov chain by transforming the standard Markov chain, which is obtained using the
Metropolis-Hastings method. We call this novel transformation {\em pseudo-lifting}.
We apply our method to graphs with geometry, or graphs with doubling dimension.
Specifically, the convergence time of our algorithm (equivalently, the mixing time
of our Markov chain) is proportional to the diameter of the network graph and hence
optimal. As a byproduct, our result provides the fastest mixing Markov chain given
the network topological constraints, and should naturally find their
applications in the context of distributed optimization, estimation and control.
\end{abstract}


\end{frontmatter}

\section{Introduction}

The recently emerging network paradigms such as sensor networks,
peer-to-peer networks and surveillance networks of unmanned vehicles have led
to the requirement of designing distributed, iterative and efficient algorithms
for estimation, detection, optimization and control.
Such algorithms provide scalability and robustness necessary for
the operation of such highly distributed and dynamic networks.
In this paper, motivated by applications of linear estimation in sensor
networks \cite{KDG03,BGPS06,MFHH02,T84}, information exchange in peer-to-peer
networks \cite{MCSY03, MS08} and reaching consensus in unmanned vehicles \cite{JLS03},
we consider the problem of computing the average of numbers in
a given network in a distributed manner. Specifically, we
consider the class of algorithms for computing the average using distributed
linear iterations. In applications of interest, the rate of convergence of
the algorithm strongly affects its performance. For example, the rate of
convergence of the algorithm determines the agility of a distributed estimator
to track the desired value \cite{BGPS06} or the error in the distributed optimization
algorithm \cite{NO08}. For these reasons, designing algorithms with fast rate of
convergence is of a great recent interest \cite{BGPS06,TN06,W07}
and the question that we consider in this paper.

A network of $n$ nodes whose communication graph is denoted by $G = (V,E)$, where
$V = \{1,\dots, n\}$ and $E = \{(i,j) : i \mbox{~and~} j \mbox{~can communicate}\}$.
Each node has a distinct value and our interest is designing a distributed iterative
algorithm for computing the average of these numbers at the nodes. A popular approach,
started by Tsitsiklis \cite{T84},
involves finding a non-negative valued $n\times n$ matrix $P = [P_{ij}]$ such that
\begin{itemize}
\item[(a)] $P$ is graph conformant, i.e. if $(i,j) \notin E$ then $P_{ij} = 0$,
\item[(b)] $\bone^T P = \bone^T$, where $\bone = [1]$ is the (column) vector of all components $1$,
\item[(c)] $P^t \bx \to \xave\bone$
as $t\to\infty$ for any $\bx \in \R_+^n$, where $\xave = \left(\sum_{i=1}^n x_i\right)/n$.
\end{itemize}
This is equivalent to finding an irreducible, aperiodic random walk on graph
$G$ with the uniform stationary distribution.

The quantity of interest, or the performance of algorithm, is the time it takes for the algorithm
to get close to $\xave \bone$ starting from any $\bx$. Specifically, given $P$, define
the $\ep$-computation time of the algorithm as
\beq
T_\ep(P)& = & \inf\left\{t~:~\forall \bx \in \R_+^n,
\frac{\|P^t\bx - x_{ave} \bone\|_\infty}{x_{ave}} \le \ep \right\}.
\eeq
It is well-known that $T_\ep(P)$ is proportional\footnote{Lemma \ref{lemma:running_mixing} states the
precise relation. Known terms, such as mixing time, that are used here
are defined in Section \ref{subsec:two-one}.} to the mixing time, denoted as $\cH(P)$, of the random walk
with transition matrix $P$. Thus, the question of interest in this paper is to
find a graph conformant $P$ with the smallest computation time or equivalently
a random walk with the smallest mixing time. Indeed, the question of designing a random
walk on a given graph with the smallest mixing time in complete generality is a well
known unresolved question.

The standard approach of finding such a $P$ is based on the method of Metropolis \cite{MRRTT53}
and Hastings \cite{H70}.
This results in a {\em reversible}
random walk $P$ on $G$. The mixing time $\cH(P)$
is known to be bounded as
$$ \frac{1}{\Phi(P)} \leq \cH(P) \leq O\left(\frac{\log n}{\Phi^2(P)}\right),$$
where $\Phi(P)$
denotes the conductance of $P$. Now, for {\em expander}
graphs the resulting $P$ induced by the Metropolis-Hastings method
is likely to have $\Phi(P) = \Theta(1)$ and hence
the mixing time is $O(\log n)$ which is {\em essentially} the fastest possible. For
example, a random walk $P=[1/n]$ on the complete graph has $\Phi(P) = 1/2$ with mixing time
$O(1)$. Thus, the question of interest is reasonably resolved for graphs
that are expanding.

Now the graph topologies arising in practice, such as those in wireless
sensor network deployed in some geographic area \cite{BGPS06, W07} or a nearest neighbor
network of unmanned vehicle \cite{ks-fb-ef:05j},
do possess {\em geometry} and are far from being expanders. A
simple example of graph with geometry is the {\em ring} graph of
$n$ nodes as shown in Figure \ref{fig0}. The Metropolis-Hastings method will lead to $P_1$ shown
in Figure \ref{fig0}(a). Its mixing time is $O(n^2 \log n)$ and
no smaller than $\Omega(n^2)$ (e.g. see \cite{BDX04}). More generally, the mixing time of any reversible random
walk on the ring graph is lower bounded by $\Omega(n^2)$ \cite{BDSX06}
for its mixing time.
Note that the diameter of the ring graph is $n$ and
obviously no random walk can mix faster than the diameter. Hence, apriori it
is not clear if the fastest mixing time is $n^2$ or $n$ or
something in between: that is, does the smallest mixing time of the random
walk on a typical graph $G$ scale like the diameter of $G$, the square of the diameter
or a power of the diameter in $(1,2)$?

\begin{figure}[htb]
\begin{center}
\centerline{\psfig{figure=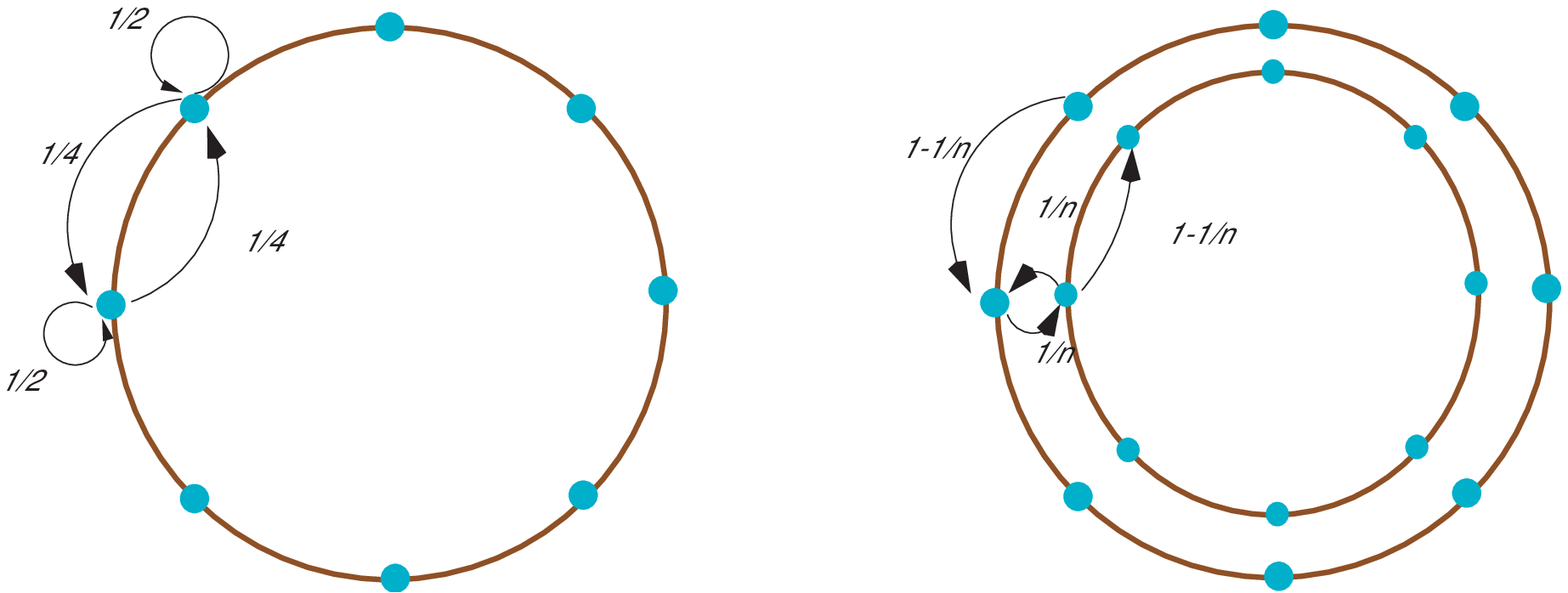,width=12cm,angle=0}}
\caption{$~~~~$(a) : $P_1$ on the ring graph $G_1$. $~~~~~~~~~~~~~~~~$ (b) : $P_2$ on the lifted ring graph $G_2$.  }
\label{fig0}
\end{center}
\end{figure}

In general, in most cases of interest the mixing time of the reversible
walk $P$ scales like $1/\Phi^2(P)$. The conductance $\Phi(P)$ relates
to diameter $D$ of a graph $G$ as $1/\Phi(P) \geq D$. Therefore,
in such situations the mixing time of random walk based on the
Metropolis-Hastings method is likely to scale like $D^2$, the square
of the diameter. Indeed, Diaconis and Saloff-Coste \cite{DS06}
established that for a certain class of graphs with {\em geometry}
the mixing time of any reversible random walk scales like at
least $D^2$ and it is achieved by the Metropolis-Hastings'
approach. Thus, reversible random walks result in rather poor performance
for graphs with geometry i.e. their mixing time is far from our best hope,
the diameter $D$.

Motivated by this, we wish to undertake the following reasonably
ambitious question in this paper: {\em is it possible to design
a random walk with mixing time of the order of diameter $D$ for
any graph?} We will answer this question in affirmative by producing
a novel construction of non-reversible random walks on the {\em lifted}
version of graph $G$. And thus, we will design iterative averaging
algorithms with the fastest possible rate of convergence.

\vspace{-.1in}
\subsection{Related work}
\vspace{-.1in}
In an earlier work, Diaconis, Holmes and Neal \cite{DHN97}
introduced a construction of a non-reversible random walk
on the ring (and more generally ring-like) graph. This random walk runs
on the {\em lifted} ring graph, which is described as $G_2$ in Figure \ref{fig0}(b).
Here, by lifting we mean
making additional copies of the nodes of the original graph and adding
edges between some of these copies while preserving the original
graph topology. Figure \ref{fig0}(b) explains the construction in \cite{DHN97}
for the ring graph. Note that each node has two copies and the lifted
graph is essentially composed of two rings: an inner ring and an outer ring.
The transition on the inner circle forms a clockwise circulation and
the transition on the outer circle forms a counterclockwise circulation.
And the probability of changing from the inner circle to the outer circle
and vice versa are $1/n$ each time. By defining transitions in this way,
the stationary distribution is also preserved; i.e.
the sum of stationary distributions of copies is equal to the stationary distribution of
their original node.
Somewhat surprisingly, the authors \cite{DHN97}
proved that this non-reversible random walk has the linear mixing time $O^*(n)$.\footnote{For a function $f:\mathbb{N}\rightarrow \mathbb{R}^+$,
$O^*(f(n)):=O(f(n)\mbox{poly} (\log n))$.} Thus, effectively (i.e. up to
$\log n$ factor) the mixing time is of the order of the diameter $n$. It should be
noted that because lifting preserves the graph topology and the stationary distribution,
it is possible
to simulate this lifted random walk on the original graph by expanding
the {\em state} appropriately, with the desired output. Equivalently,
it is possible to use a lifted random walk for linear averaging by running
iterations with extra states.\footnote{The details are given in Section \ref{sec:five}.}

The following question arose from the work of \cite{DHN97}: given graph $G$
and random walk $P$ on $G$, is it possible to design a non-reversible random walk
on the lifted version of $G$ which mixes subsequently faster than $P$?
Can it mix in $O(D)$? This question was addressed in a subsequent work by Chen, Lov\'{a}sz and
Pak \cite{CLP02}. They provided an explicit construction of a random
walk on a lifted version of $G$ with mixing time $O^*(1/\Phi(P))$. Further,
they showed that, under the notion of lifting (implicity) introduced
by \cite{DHN97} and formalized in \cite{CLP02}, it is not possible to design
such a lifted random walk with mixing time smaller than $\Omega(1/\Phi(P))$.

Now note that $1/\Phi(P)$ can be much larger than the diameter $D$. As a simple
example, consider a ring graph with $P$ exactly the same as that in
Figure \ref{fig0}(a), but with a difference that for two edges the
transition probabilities are $\delta(n)$ instead of $1/4$ (and the
transition probabilities of endpoints of these edges appropriately adjusted).
Then, it can be checked that $1/\Phi(P)$ is $\Omega(n/\delta(n))$
which can be arbitrarily poor compared to the diameter $n$ by choosing
small enough $\delta(n)$. A more interesting example showing this poorer
scaling of $1/\Phi(P)$ compared to diameter, even for the Metropolis-Hastings'
style construction, is presented in Section \ref{sec:three} in the
context of a ``Barbell graph'' (see Figure \ref{fig3}). Thus, the lifting approach of \cite{DHN97, CLP02}
can not lead to a random walk with mixing time of the order of diameter and
hence the question of existence or design of such a random walk remains
unresolved.

As noted earlier, the lifted random walk can be used to design iterative
algorithms (for computing an average) on the original graph since the topology
of the lifted graph and the stationary distribution of the lifted random walk
``projects back'' onto those of the original graph and the random walk respectively.
However, running algorithm based on lifted random walks on the original graph
requires additional states. Specifically, the lifted random walk based algorithm
can be simulated on the original graph by running multiple {\em threads} on each
node. Specifically, the number of operations performed per iteration across
the network depends on the size\footnote{In this paper, the size of a random walk (resp. graph)
is the number of non-zero entries in its transition matrix (resp. number of edges in the graph).}
of the lifted walk (or graph).  In the construction of \cite{CLP02} for a general graph,
this issue about the size of the lifted walk was totally ignored as the authors' interest was
only the time complexity, not the size. Therefore, even though time may
reduce under the construction of \cite{CLP02} the overall cost (equal to the product
of time and size) may not be reduced; or even worse, it may increase.

Therefore, from the perspective of the application of iterative algorithms we
need a notion of lifting that leads to a design of a random walk that has (a)
mixing time of the order of diameter of the original graph and (b) the smallest
possible size.


\subsection{Our contributions}

In this paper, we answer the above stated question affirmatively. As noted
earlier, the notion of lifting of \cite{DHN97, CLP02} can not help in
answering this question. For this reason, we introduce a notion of
{\em pseudo-lifting} which can be thought of as a relaxation of the notion
of lifting. Like lifting, the notion of pseudo-lifting preserves
the topological constraints of the original graph. But the relaxation
comes in preserving the stationary distribution in an approximate
manner. However, it should be noted that is still possible to use the pseudo-lifted random
walk to perform the iterative algorithm without any approximation
errors (or to sample objects from a stationary distribution
without any additional errors) since the stationary distribution
of pseudo-lifting under a restricted projection provides the original stationary
distribution exactly. Thus, operationally our notion of pseudo-lifting is as effective
as lifting.

First, we use pseudo-lifting to design a random walk with mixing time of the
order of diameter of a given graph with the desired stationary distribution. To
achieve this, we first use the Metropolis-Hastings method to construct a random
walk $P$ on the given graph $G$ with the desired stationary distribution.
Then, we {\em pseudo-lift} this $P$ to obtain a random walk with mixing time of the
order of diameter of $G$. This approach is stated as Theorem \ref{thm:mixing_pseudo}. 

As discussed earlier, the utility of such constructions lies in the context
of graphs with {\em geometry}. The graphs with (fixed) finite {\em doubling dimension},
introduced in \cite{A83,H01,GKL03, dasgupta-khandekar-stoc08},
serve as an excellent model for such a class of graphs. Roughly speaking,
a graph has doubling dimension $\rho$ if the number of nodes within the
shortest path distance $r$ of any node of $G$ is $O(r^\rho)$
(i.e. polynomial growth of the neighborhood of a node). We apply our construction
of pseudo-lifting to graphs with finite doubling dimension $\rho$ to
obtain a random walk with mixing time of the order of diameter $D$. In
order to address the concern with expansion in the size of the
pseudo-lifted graph, we use the geometry of the original graph explicitly. Specifically,
we reduce the size of the lifted graph by a clever combination of
clustering, geometry and pseudo-lifting. This formal result is stated
as follows and its proof is in Section \ref{subsec:four-two}.

\begin{theorem}\label{thm:performance_pseudo_lifting}
Consider a connected graph $G$ with doubling dimension $\rho$ and diameter
$D$. It is possible to explicitly construct a pseudo-lifted random walk on $G$
with mixing time $O(D)$ chain and size $O\left(Dn^{1-\frac{1}{1+\rho}}\right)$.
\end{theorem}

As a specific example, consider a $d$-dimensional grid whose doubling dimension is $d$.
The Metropolis-Hastings
method has mixing time $\Omega\left(n^{2/d}\right)$, compared to our construction
with mixing time $O\left(n^{1/d}\right)$. Further, our construction leads to an increase
in size of the random walk only by $O\left(n^{1/d(d+1)}\right)$ factor.
That is, pseudo-lifting is optimal in terms of the number of iterations, which is equal
to diameter, and in terms of cost per iteration it is lossy by a relatively small
amount, for example $O\left(n^{1/d(d+1)}\right)$ for the $d$-dimensional grid.

In general, we can use pseudo-lifting to design iterative algorithms
for computing the average of given numbers on the original graph itself. We
describe a precise {\em implementation} of such an algorithm in Section
\ref{sec:five}. The use of pseudo-lifting, primarily effective for a class of
graphs with geometry, results in the following formal result whose proof
is in Section \ref{subsec:five-two}.

\begin{theorem}\label{thm:performance_averaging}
Consider a given connected graph $G$ with diameter $D$ and each
node with a distinct value. Then, (using a pseudo-lifted random
walk) it is possible to design an iterative algorithm whose $\ep$-computation
time is $T_\ep =O^*\left(D\log\frac1{\varepsilon}\right)$.
Further, if $G$ has doubling dimension $\rho$, then the network-wide
total number of operations (essentially, additions) per iteration of the
algorithm is $O\left(D n^{1-\frac1{\rho+1}}\right)$.
\end{theorem}

As a specific example, recall a $d$-dimensional grid
with doubling dimension $d$ and diameter $n^{1/d}$.
The Metropolis-Hastings method will have mixing time $\Omega\left(n^{2/d}\right)$
and per iteration number of operations $\Theta(n)$. Therefore, the number of total
operations is $O\left(n^{1+\frac{2}{d}}\right)$ (even the randomized gossip algorithm
of \cite{BGPS06} will have this total cost). Compared to this,
Theorem \ref{thm:performance_averaging} implies the number of iterations
would be $O\left(n^{1/d}\right)$ and per iteration cost would be $O\left(n^{1+\frac{1}{d(d+1)}}\right)$.
Therefore, the total cost is $O\left(n^{1+\frac{d+2}{d(d+1)}}\right)$ which is
essentially close to $O\left(n^{1+1/d}\right)$ for large $d$. Thus, it strictly improves
performance over the Metropolis-Hastings method by roughly $n^{1/d}$ factor.
It is worth nothing that no algorithm can have the number of total operations less
than $\Omega\left(n^{1+1/d}\right)$ and the number of iterations less than $\Omega\left(n^{1/d}\right)$.

\vspace{.1in}

For the application of interest of this paper, it was necessary to introduce a
new notion of lifting and indeed we found one such notion, i.e. pseudo-lifting.
In general, it is likely that for certain other applications such a notion
may not exist. For this reason, we undertake the question of designing a lifted
(not pseudo-lifted) random walk with the smallest possible size
since the size (as well as the mixing time) decides the cost of the algorithm that
uses lifting. Note that the average-computing algorithm in Section \ref{sec:five} can
also be implemented via lifting instead of pseudo-lifting,
and  the size of lifting leads to the total number of operations\footnote{One can derive its explicit performance bound as Theorem \ref{thm:performance_averaging}. It turns out that lifting is worse than pseudo-lifting in its performance, but it is more robust in its construction.}. As the first
step, we consider the construction of Chen, Lov\'{a}sz and Pak \cite{CLP02}.
We find that it is rather {\em lossy} in its size. Roughly speaking,
their construction tries to build a logical {\em complete} graph topology
using the underlying graph structure. In order to construct one of $n^2$
edges of this complete graph topology, they use a solution of a flow
optimization problem. This solution results in multiple
paths between a pair of nodes. Thus, in principle, their approach can lead to
a very large size. In order to reduce this size, we use two natural
ideas: one, use a sparse expander graph instead of the complete graph and two, use a solution
of unsplittable flows \cite{KS02}. Intuitively, this approach seems reasonable
but in order to make it work, we need to overcome rather
non-trivial technical challenges. To address these challenges, we develop
a method to analyze hybrid non-reversible random walks, which should be of
interest in its own right. The formal result is stated as follows
and see Section \ref{sec:six} for its complete proof.
\begin{theorem}\label{thm:performance_expander_lifting}
Consider a given connected graph $G$ with a random walk $P$. Then, there exists
a lifted random walk with mixing time $O^*(1/\Phi(P))$ and
size $O^*(|E(P)|/\Phi(P))$, where
$$ E(P) = \{(i,j): P_{ij} \neq 0 ~\mbox{or}~P_{ji}\neq 0\}.$$
\end{theorem}
Note that the lifted random walk in \cite{CLP02} has
size $\Omega(n^2/\Phi(P))$, hence
our lifting construction leads to the reduction of its size by $\Theta(n)$ factor
when $G$ is sparse\footnote{A graph $G=(V,E)$ is sparse if $|E|=O(|V|)$.}.
Finally, we note that the methods developed for understanding
the expander-based construction (and proof of Theorem \ref{thm:performance_expander_lifting})
can be useful in making pseudo-lifting more {\em robust}, as discussed
in the Section \ref{sec:conc}.


\section{Preliminaries and Backgrounds}\label{sec:two}

\subsection{Key notions and definitions}\label{subsec:two-one}
In this paper, $G=(V,E)$ is a given graph with $n$ nodes i.e. $|V|=n$.
We may use $V(G)$ to represent vertices of $V$ of $G$.
$P$ always denotes a transition matrix of a graph conformant random walk (or Markov chain) on $G$ with
its stationary distribution $\pi$
i.e. $P_{ij} > 0 $ only if $(i,j) \in E$, and $\pi^TP=\pi^T$.
We will use the notion of ``Markov chain'' or ``random walk'' depending on
which notion is more relevant to the context.
The reverse chain $P^*$ of $P$ is
defined as: $P^*_{ij}=\pi_j P_{ji}/\pi_i$ for all $(i,j) \in E$.
We call $P$ {\em reversible}  if $P=P^*$.
Hence, if $\pi$ is uniform\footnote{$\pi$ is uniform when $\pi_i=1/n,\forall i$.}, $P$ is a symmetric matrix.
The conductance of $P$ is defined as
$$\Phi(P)=\min_{S\subset V}\frac{\sum_{i\in S,j\in
V\backslash S}\pi_i P_{ij}}{\pi(S)\pi(V\backslash S)},$$ where $\pi(A)=\sum_{i\in A}\pi_i$.

Although there are various (mostly equivalent) definitions of Mixing time
that are considered in the literature based on different measures of
the distance between distributions,
we primarily consider the definition of Mixing time
from the stopping rule. A stopping rule $\Gamma$ is a
stopping time based on the random walk of $P$:  at any
time, it decides whether to stop or not, depending on the walk seen
so far and possibly additional coin flips.  Suppose, the starting
node $w^0$ is drawn from distribution $\sigma$. The distribution of
the stopping node $w^{\Gamma}$ is denoted by $\sigma^{\Gamma}=\tau$
and call $\Gamma$ as a stopping rule from $\sigma$ to $\tau$. Let
$\cH(\sigma,\tau)$ be the infimum of mean length over all such
stopping rules from $\sigma$ to $\tau$. This is well-defined as
there exists the following stopping rule from $\sigma$ to $\tau$:
{\em select $i$ with probability $\tau_i$ and walk until getting to
$i$}. Now, we present the definition of
the (stopping rule based) Mixing time
$\mathcal{H}$.

\begin{definition}[Mixing time]
$\mathcal{H}=\max_{\sigma} \mathcal{H}(\sigma,\pi).$
\end{definition}
Therefore, to bound $\mathcal{H}$, we need to design a stopping rule
whose distribution of stopping nodes is $\pi$.
\vspace{-.1in}
\subsection{Metropolis-Hastings method}
\vspace{-.1in}
The Metropolis-Hastings method (or Glauber dynamics \cite{kenyon01glauber}) has been extensively studied
in recent years due to its local constructibility.
For a given graph $G=(V,E)$ and distribution $\pi$ on $V$,
the goal is to produce a random walk $P$ on $G$ whose stationary distribution is $\pi$.
The underlying idea of the random walk produced by this method is
choosing a neighbor $j$ of the current vertex $i$ at uniformly random and
moving to $j$ depending on the ratio between $\pi_i$ and $\pi_j$. Hence, its explicit
transition matrix $P$ is following:
$$P_{ij} =
\begin{cases}
    \frac1{2d}\min\{\frac{\pi_j}{\pi_i},1\} & \text{if }(i,j)\in E\\
    0 & \text{if }(i,j)\notin E ~ \text{and }i\neq j\\
    1-\sum_{k\neq i} P_{ik} & \text{if }i= j\\
  \end{cases},$$
where $d_i$ is a degree of vertex $i$ and $d=\max_i d_i$. It is easy to
check that $\pi^T P=\pi^T$ and $P$ is reversible.
\vspace{-.1in}

\subsection{Lifting}\label{subsec:two-three}
\vspace{-.1in}
As stated in the introduction,
motivated by a simple ring example of Diaconis et al. \cite{DHN97},
Chen et al. \cite{CLP02} use the following notion of lifting.

\begin{definition}[Lifting]
A  random walk $\hat{P}$ on graph $\hat{G}=(\hat{V}, \hat{E})$
is called a lifting of random walk $P$ on graph $G=(V,E)$ if
there exists a many-to-one function $f : \hat{V} \to V$
such that the following holds: (a) for any $\hat{u}, \hat{v} \in \hat{V}$,
$(\hat{u}, \hat{v}) \in \hat{E}$ only if $(f(\hat{u}),f(\hat{v})) \in E$;
(b) for any $u,v\in V$, $\pi(u) = \widehat{\pi}(f^{-1}(u))$ and
$Q(u,v) =\widehat{Q}(f^{-1}(u), f^{-1}(v))$. Here $Q$ and $\hat{Q}$
are ergodic flow
matrices for $P$ and  $\hat{P}$ respectively.
\end{definition}
Here, the ergodic flow matrix $Q = [Q_{ij}]$ of $P$ is defined as: $Q_{ij} = \pi_iP_{ij}$.
It satisfies: $\sum_{i,j}Q_{ij}=1$, $\sum_i Q_{ij}=\sum_i Q_{ji}$
and $\sum_i Q_{ij}=\pi_j$. Conversely, every non-negative
matrix $Q$ with these properties defines a random walk with
the stationary distribution $\pi$.
In this paper, $\widehat{P}$ means a lifted (or pseudo-lifted) random walk of $P$. Similarly
$\widehat{G}$, $\widehat{V}$, $\widehat{E}$ and $\widehat{\pi}$ are the lifted (or pseudo-lifted) versions
of their original one.

Chen et al. \cite{CLP02} provided an
explicit construction to lift a given general random walk $P$ with
almost optimal speed-up in terms of mixing time. Specifically, they obtained the following result.
\begin{theorem}[\cite{CLP02}]\label{thm:CLP}
For a given random walk $P$, it is possible to explicitly construct a lifted random
walk of $P$ with mixing time $O^*(1/\Phi(P))$. Furthermore,
any lifted random walk of $P$ needs at least $\Omega(1/\Phi(P))$ time to mix.
\end{theorem}

\vspace{-.1in}
\subsection{Auxiliary backgrounds}\label{subsec:two-four}
\vspace{-.1in}
\subsubsection{$\ep$-Mixing time}
\vspace{-.1in}

Here we introduce a different (and related) notion of Mixing time which measures more
explicitly how fast the random walk converges to the stationarity.
The following notions, $\tau(\beps), \tau_2(\beps)$ are related to ${\cal H}$.
This relation can be found in detail in the survey by Lov\'{a}sz and Winkler \cite{LW98}.
For example, we will use this relation explicitly in Lemma \ref{lemma:running_mixing}.

Now we define these related definitions of mixing time. To this end, as before consider
a random walk $P$ on a graph $G=(V,E)$. Let $P^t(x,\cdot)$ denote
the distribution of the state after $t$ steps under $P$,
starting from an initial state $x \in V$. For the random walk of our interest,
$P^t(x,\cdot)$ goes to $\pi$ as $t\rightarrow\infty$. We present the
definitions based on the total variation distance and the $\chi^2$-distance.
\begin{definition}[$\ep$-Mixing time]
Given $\beps > 0$, let $\tau(\beps)$ and $\tau_2(\beps)$ represent
$\beps$-Mixing time of the random walk with respect to the total variation
distance and the $\chi^2$-distance respectively. Then, they are
{\small
\begin{align*}
&\tau(\beps)=\min\lf\{t:\forall x\in V, \frac{1}{2} {\sum_{y\in \Omega}
\left|P^t(x,y)-\pi(y)\right|}\leq \beps \rf\},\\
&\tau_2(\beps)=\min\lf\{t:\forall x\in V,\sqrt{\sum_{y\in
\Omega} \frac1{\pi(y)}\left(P^t(x,y)-\pi(y)\right)^2}\leq \beps \rf\}.
\end{align*}
}
\end{definition}
\vspace{-.1in}
\subsubsection{Additional Techniques to bound Mixing Times}
\vspace{-.1in}

Various techniques have been developed over past three decades or
so to estimate Mixing time of a given random walk.
The relation between the conductance and the mixing time in the introduction is one of them.
We review some of
the key other techniques that will be relevant for this paper.

\vspace{.1in}
\noindent {\it Fill-up Lemma.}
Sometimes, due to the difficulty for
designing such an exact stopping rule, we use the following strategy
for bounding the mixing time $\mathcal{H}$.

\begin{itemize}
\item[] {\em Step 1.} For a positive constant $\beps$ and any starting distribution $\sigma$,
we design a stopping rule whose
stopping distribution $\gamma$ is $\beps$-far from $\pi$ (i.e. $\gamma\geq(1-\varepsilon)\pi$).
This gives the upper bound for $H(\sigma,\gamma)$.
\item[] {\em Step 2.} We bound $\mathcal{H}$ by $H(\sigma,\gamma)$ using
the following fact known as {\em fill-up  Lemma} in \cite{A82}:
$$\mathcal{H}\leq \frac1{1-\varepsilon}
\mathcal{H}_{\underline{\varepsilon}},$$
where
$\mathcal{H}_{\underline{\varepsilon}}=\max_{\sigma} \min_{\gamma\geq(1-\varepsilon)\pi} \mathcal{H}(\sigma,\gamma)$.
\end{itemize}

\vspace{.1in}
\noindent {\it Eigenvalue.} If $P$ is reversible, one can view
$P$ as a self-adjoint operator on a suitable inner product space and
this permits us to use the well-understood spectral theory of
self-adjoint operators. It is well-known that $P$ has $n=|V|$ real
eigenvalues $1=\lambda_0
>\lambda_1 \geq \lambda_2\geq \cdots \geq \lambda_{n-1}> -1$. The
$\beps$-mixing time $\tau_2(\beps)$ is related as
$$\tau_2(\beps)\leq \left\lceil \frac1{\lambda_{P}}\log
\frac1{\beps \sqrt{\pi_0}} \right\rceil,$$ where
$\lambda_{P}=1-\max\{|\lambda_1|,|\lambda_{n-1}|\}$
and $\pi_0=\min_i\pi_i$. The
$\lambda_{P}$ is also called the spectral gap of $P$.
When $P$ is non-reversible, we consider $PP^{*}$. It is easy
to see that the Markov chain with $PP^*$ as its transition matrix
is reversible. Let $\lambda_{PP^*}$ be the spectral gap of this
reversible Markov chain. Then, the mixing time of the original Markov
chain (with its transition matrix $P$) is bounded above as:
\begin{equation}\label{rel:mixing_eigenvalue}
\tau_2(\beps)\leq \left\lceil \frac2{\lambda_{PP^*}}\log
\frac1{\beps \sqrt{\pi_0}} \right\rceil.
\end{equation}


\section{Pseudo-Lifting}\label{sec:three}

Here our aim is to obtain a random walk with mixing time of the order of the diameter for a given
graph $G$ and stationary distribution $\pi$. As explained in the introduction,
the following approach based on lifting does not work for this aim: first obtain a random walk with the
desired stationary distribution using the Metropolis-Hastings method, and then
lift it using the method in \cite{CLP02}.

For example,
consider the {\em
Barbell graph} $G$ as shown in Figure \ref{fig3}:
{\em two complete graphs of $n/2$ nodes
connected by a single edge}. And, suppose $\pi$ is uniform.
Now, consider a random walk $P$ produced by the Metropolis-Hastings method:
the next transition
is uniform among all the neighbor for each node. For such a random
walk, it is easy to check that $1/\Phi(P) = \Omega(n^2)$ and $\mathcal{H}=\Omega(n^4)$.
Therefore, the
mixing time of any lifting is at least $\Omega(n^2)$. However,
this random walk is ill-designed to begin with because $1/\Phi(P)$ can
be decreased up to
$O(n)$ by defining its random walk in another way (i.e.
increasing the probability of its linkage edge, and adding
self-loops to non-linkage nodes not to change its stationary
distribution).
$1/\Phi(P)$ is still far from the diameter $D = O(1)$ nevertheless. Hence,
from Theorem \ref{thm:CLP}, lifting cannot achieve $O(D)$-mixing.

\begin{figure}[htb]
\begin{center}
\centerline{\psfig{figure=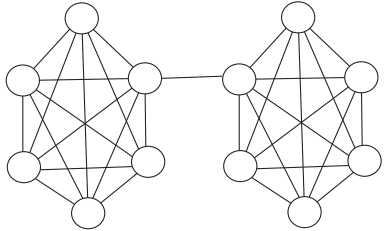,width=5cm,angle=0}}
\caption{The Barbell graph with $12$ nodes.}
\label{fig3}
\end{center}
\end{figure}

Motivated by this limitation,
we will use the following new notion of lifting, which we call pseudo-lifting, to design a
$O(D)$-mixing random walk.

\begin{definition}[Pseudo-Lifting]\label{def:pseudo-lifting}
A random walk $\hat{P}$ is called a pseudo-lifting of $P$
if there exists a many-to-one function $f : \hat{V} \to V$, $T \subset \widehat{V}$ with $|T| =
|V|$ such that the following holds: (a) for any $\hat{u}, \hat{v}
\in \hat{V}$,  $(\hat{u}, \hat{v}) \in \hat{E}$ only if
$(f(\hat{u}),f(\hat{v})) \in E$, and (b) for any $u \in V$,
$\hat{\pi}(f^{-1}(u) \cap T)=\frac12\pi(u).$\footnote{In fact, $\frac12$ can be replaced by any constant between 0 and 1.}
\end{definition}

The property (a) in the definition implies that one can simulate the pseudo-lifting $\widehat{P}$
in the original graph $G$. Furthermore, the property (b) suggests that (by concentrating on the set $T$), it
is possible to simulate the stationary distribution $\pi$ exactly
via pseudo-lifting.
Next we present its
construction.
\vspace{-.1in}

\subsection{Construction}\label{subsec:three-one}
\vspace{-.1in}

For a given random walk $P$, we will construct the pseudo-lifted random walk $\widehat{P}$ of $P$.
It may be assumed that $P$ is given by the Metropolis-Hastings method.
We will construct the pseudo-lifted graph $\widehat{G}$ by adding vertices and edges to $G$, and
decide the values of the ergodic flows $\widehat{Q}$ on $\widehat{G}$, which defines its
corresponding random walk $\widehat{P}$.

First, select an arbitrary node
$v$. Now, for each $w \in V$, there exist paths $\P_{wv}$ and
$\P_{vw}$, from $w$ to $v$ and $v$ to $w$ respectively.  We will assume
that all the paths are of length $D$: this can be achieved by
repeating same nodes. Now, we construct a pseudo-lifted graph
$\hat{G}$ starting from $G$.

First, create a new node $v^{\prime}$ which is a copy of the chosen vertex $v$.
Then, for every node $w$,
add directed paths $\P^{\prime}_{wv}$, a copy of $\P_{wv}$, from $w$ to $v^\prime$. Similarly,
add $\P^{\prime}_{vw}$ (a copy of $\P_{vw}$) from $v^{\prime}$ to $w$. Each
addition creates $D-1$ new interior nodes. Thus, we have essentially created
a {\em virtual star topology} using the paths of the old graph by adding $O(nD)$
new nodes in total. (Every new node is a copy of an old node.)

Now, we define the ergodic flow matrix $\widehat{Q}$  for this graph $\hat{G}$ as follows: for
an edge $(i,j)$,
\begin{align*}
\widehat{Q}_{ij}=\begin{cases}
\frac{\delta_1}{2D}\pi_w\ \ &\text{if}\ (i,j) \in E(\P^{\prime}_{wv})\text{ or }E(\P^{\prime}_{vw})\\
(1-\delta_1)Q_{ij} &\text{if}\ (i,j) \in E(G),
\end{cases}
\end{align*}
where $\delta_1\in [0.1]$ is a constant we will decide later in (\ref{eq:delta1}).
It is easy to check that $\sum_{ij} \widehat{Q}_{ij}=1,\sum_{j}
\widehat{Q}_{ij}=\sum_{j} \widehat{Q}_{ji}$.  Hence it defines a
a random walk on $\hat{G}$. The stationary distribution of this
pseudo-lifting is
\begin{align*}
\widehat{\pi}_i=\begin{cases}
\frac{\delta_1}{2D}\pi_w\ \ &\text{if}\ i\in (V(\P^{\prime}_{wv}) \cup V(\P^{\prime}_{vw})) \backslash \{w,v^{\prime}\}\\
\lf(1-\delta_1 + \frac{\delta_1}{2D}\rf)\pi_i \ \ &\text{if}\ i\in V(G)\\
\frac{\delta_1}{2D}\ \ &\text{if}\ i=v^{\prime}
\end{cases}
\end{align*}
Given the above definition of $\hat{Q}$ and corresponding stationary
distribution $\hat{\pi}$, it satisfies the requirements of pseudo-lifting in Definition \ref{def:pseudo-lifting}
if we choose $\delta_1$ such that
\begin{equation}
1/2= \delta_1 \lf(1-\frac{1}{2D}\rf),\label{eq:delta1}
\end{equation} and
$T=V(G)$; i.e. $T$ is the set of old nodes.
\vspace{-.1in}

\subsection{Mixing time}\label{subsec:three-two}
\vspace{-.1in}

We claim the following bound
on the mixing time of the pseudo-lifting we constructed.
\begin{theorem}\label{thm:mixing_pseudo}
The mixing time of the random walk $\widehat{P}$ defined by $\widehat{Q}$ is
$O(D)$.
\end{theorem}
\begin{proof}
We will design a stopping rule where the distribution of the stopping node is $\widehat{\pi}$,
and analyze its expected length. At first, walk until visiting
$v^{\prime}$, and toss a coin $X$ with the following probability.
{\small
\begin{align*}
X=\begin{cases} 0\ &\text{with probability}\
\frac{\delta_1}{2D}\\
1\ \ &\text{with probability}\
\frac{\delta_1 (D-1)}{2D}\\
2\ \ &\text{with probability}\
1-\delta_1+\frac{\delta_1}{2D}\\
3\ \ &\text{with probability}\
\frac{\delta_1 (D-1)}{2D}\\
\end{cases}
\end{align*}}
Depending on the value of $X$, the stopping node is decided as follows.
\begin{itemize}
\item[$\circ$] $X=0$ : {\em Stop at $v^{\prime}$}.  The probability for stopping at $v^{\prime}$ is $\Pr[X=0]=\frac{\delta_1}{2D}$, which is exactly $\widehat{\pi}_{v^{\prime}}$.
\item[$\circ$] $X=1$ : {\em Walk a directed path $P^{\prime}_{vw}$, and choose an interior node
of $P^{\prime}_{vw}$ uniformly at random, and stop there}.
For a given $w$, the probability for walking $P^{\prime}_{vw}$ is easy to check $\pi_w$.
There are $D-1$ many interior nodes, hence, for an interior node $i$ of $P^{\prime}_{vw}$,
the probability for stopping at $i$ is $$\Pr[X=1]\times \pi_w \times \frac1{D-1}=
\frac{\delta_1}{2D}\pi_w=\widehat{\pi}_i.$$
\item[$\circ$] $X=2$ : {\em Stop at the end node $w$ of $P^{\prime}_{vw}$}.
The probability for stopping at $w$ is $$\Pr[X=2]\times \Pr[\text{walk }P^{\prime}_{vw}]=
\lf(1-\delta_1+\frac{\delta_1}{2D}\rf)\times \pi_w=\widehat{\pi}_w.$$
\item[$\circ$] $X=3$ : {\em Walk until getting a directed path
$P^{\prime}_{wv}$, and choose an interior node of $P^{\prime}_{wv}$
uniformly at random, and stop there.}
Until getting a directed path
$P^{\prime}_{wv}$, the pseudo-lifted random walk defined by $\widehat{Q}$ is same as
the original random walk. Since the distribution $w \in V(G)$ of the walk
at the end of the previous step is exactly $\pi$,
it follows that the distribution $\pi$ over the nodes of $V(G)$
is preserved under this walk till walking on $P^{\prime}_{wv}$.
From the same calculation as the case $X=1$, the probability of
stopping at the interior node $i$ of $P^{\prime}_{wv}$ is $\widehat{\pi}_i$.
\end{itemize}
Therefore, we have established the existence of a stopping rule that takes an
arbitrary starting distribution to the stationary distribution $\hat{\pi}$. Now, this stopping
rule has an average length $O(D/\delta_1)$: since the probability of getting on a
directed path $P^{\prime}_{wv}$ at $w$ is $\frac{\delta_1}{2D}/(1-\delta_1+\frac{\delta_1}{2D})
=\Theta(\delta_1/D)$,  the expected
numbers of walks until visiting $v^{\prime}$ and getting a directed path when $X=3$ are
$O(D/\delta_1)=O(D)$ from (\ref{eq:delta1}) in both cases. This completes the proof.
\end{proof}

\vspace{-.1in}

\section{Pseudo-Lifting: use of geometry}\label{sec:four}

\vspace{-.1in}

The graph topologies arising in practice, such as those in wireless
sensor network deployed in some geographic area or a nearest neighbor
network of unmanned vehicles \cite{ks-fb-ef:05j},
do possess {\em geometry} and are far from being expanders.
A good model for graphs with geometry is a class of graphs with finite doubling dimension
which is defined as follows.

\begin{definition}[Doubling Dimension]
Consider a metric space $\cM =
(\cX, \bd)$, where $\cX$ is the set of point endowed with a metric
$\bd$.  Given $x \in \cX$, define a ball of radius $r \in \Rp$
around $x$ as $\bB(x, r) = \{ y \in \cX : \bd(x, y) < r\}$. Define
$$\rho(x, r) = \inf \{ K \in \N:  \exists ~y_1,\dots, y_K \in \cX,  \bB(x,r) \subset \cup_{i=1}^K \bB(y_i, r/2) \}.$$
Then, the $\rho(\cM) = \sup_{x \in \cX, r \in \Rp} \rho(x,r)$ is
called  the {\em doubling constant} of $\cM$ and $\log_2 \rho(\cM)$
is called the {\em doubling dimension} of $\cM$. The doubling dimension of a graph
$G=(V,E)$ is defined with respect to the metric induced on $V$ by the shortest
path metric.
\end{definition}

For graphs with finite doubling dimension, we will design a pseudo-lifting with its efficient size.
Recall the basic idea for the construction of the pseudo-lifting in Section \ref{sec:three}
is creating a {\em virtual star topology} using paths from every node to a fixed root,
and the length of paths grows the size of the pseudo-lifting.
To reduce the overall length of paths, we consider clusters of nodes such that
nodes in each cluster are close to each other, and pick a sub-root node in each cluster.
And then, build a star topology in each cluster around its sub-root and
connect every sub-root to the root. This creates a {\em hierarchical star topology} (or say a {\em tree topology})
as you see the example of the line graph in Figure \ref{fig4}(b).
Since it needs paths of short length in each cluster, the overall length of paths would be decreased.

For a good clustering,
we need to decide which nodes would become sub-roots. A natural candidate for them
is the $R$-net $Y \subset V$ of a graph $G$ defined as follows.
\begin{definition}[$R$-net]
For a given graph $G=(V,E)$, $Y \subset V$ is a $R$-net if
\begin{itemize}
\item[(a)] For every $v \in V$, there exists $u \in Y$ such that the shortest
path distance between $u, v$ is at most $R$.
\item[(b)] The distance between
any two $y, z \in Y$ is more than $R$.
\end{itemize}
\end{definition}

Such an $R$-net can be found
in $G$ greedily, and as you will see the proof of Lemma \ref{lemma:pseudo_size},
the small {\em doubling dimension} of $G$
guarantees the existence of a good $R$-net for our purpose.

\begin{figure}[htb]
\begin{center}
\centerline{\psfig{figure=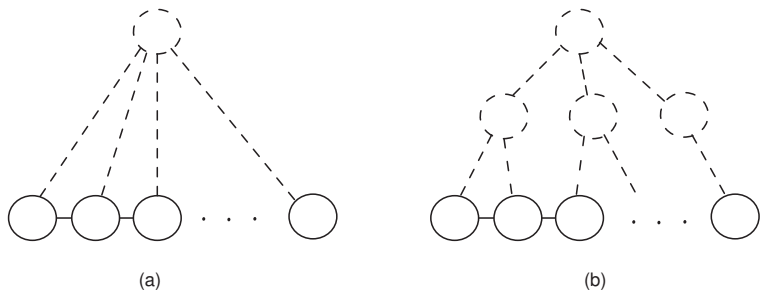,width=10cm,angle=0}}
\caption{For a given line graph with $n$ nodes,
(a) is the star topology which used in the construction of the pseudo-lifted graph
in Section \ref{subsec:three-one}. (b) is the hierarchical star topology which will be
used in this section for the new construction of pseudo-lifting.}
\label{fig4}
\end{center}
\end{figure}
\vspace{-.1in}

\subsection{Construction}\label{subsec:four-one}
\vspace{-.1in}

For a given random walk $P$, we will construct the pseudo-lifted random walk $\widehat{P}$ of $P$
using a hierarchical star topology.
Denote $\pi$ and $G=(V,E)$ be the stationary distribution and the underlying graph of $P$ again.
As the previous construction in Section \ref{subsec:three-one},
we will construct the pseudo-lifted graph $\widehat{G}$ by extending $G$, and define the ergodic flow
matrix $\widehat{Q}$ on $\widehat{G}$, which leads to its
corresponding random walk $\widehat{P}$.

Given a $R$-net $Y$, match each node $w$ to the
nearest $y \in Y$ (breaking ties arbitrarily). Let $C_y =\{w|\ w$
matched to $y\}$ for $y \in Y$. Clearly, $V = \cup_{y \in Y} C_y$.
Finally, for each $y \in Y$ and for any $w \in C_y$ we have paths
$\P_{wy}, \P_{yw}$ between $w$ and $y$ of length $R$ exactly.
Also, for each $y \in Y$, there exit $\P_{yv},\P_{vy}$ between $y$ and $v$ of length
$D$ exactly (we allow the repetition of nodes to hit this length exactly).

Now, we construct the pseudo-lifted graph $\hat{G}$. As the construction in Section \ref{subsec:three-one},
select an
arbitrary node $v \in V$ and create its copy $v^\prime$ again.
Further, for each $y \in Y$, create two copies $y_1^\prime$ and
$y_2^\prime$. Now, add directed paths $\P^{\prime}_{wy}$, a copy of
$\P_{wy}$, from $w$ to $y^{\prime}_1$ and add $\P^{\prime}_{yv}$, a
copy of $\P_{yv}$, from $y^{\prime}_1$ to $v^{\prime}$. Similarly,
add $\P^{\prime}_{vy}$ and $\P^{\prime}_{yw}$ between $v^\prime$,
$y^{\prime}_2$ and $y^{\prime}_2$, $w$. In total, this construction for $\hat{G}$ adds $2D
|Y| + 2Rn$ edges to $G$. Now, the ergodic flow matrix
$\widehat{Q}$ on $\hat{G}$ is defined as follows: for any $(i,j)$ of
$\hat{G}$,
\begin{align*}
\widehat{Q}_{ij}=\begin{cases}
\frac{\delta_2}{2(R+D)}\pi_w\ \ &\text{if}\ (i,j) \in E(\P^{\prime}_{wy})\text{ or }E(\P^{\prime}_{yw})\\
\frac{\delta_2}{2(R+D)}\pi(C_y)\ \ &\text{if}\ (i,j) \in E(\P^{\prime}_{yv})\text{ or }E(\P^{\prime}_{vy})\\
(1-\delta_2)Q_{ij} &\text{if}\ (i,j) \in E(G)
\end{cases},
\end{align*}
where $\pi(C_y)=\sum_{w\in C_y} \pi_w$ and $\delta_2\in [0.1]$ is a constant decided later\footnote{
See the equation (\ref{eq:delta2}) and check $\delta_2\approx 1/2.$}.
It can be checked that
$\sum_{ij} \widehat{Q}_{ij}=1,\sum_{j} \widehat{Q}_{ij}=\sum_{j}
\widehat{Q}_{ji}$. Hence it defines a random walk on $\hat{G}$. The
stationary distribution of this pseudo-lifted chain is
{\small \begin{align*}
\widehat{\pi}_i=\begin{cases}
\frac{\delta_2}{2(R+D)}\pi_w\ \ &\text{if}\ i\in (V(\P^{\prime}_{wy}) \cup V(\P^{\prime}_{yw})) \backslash \{w,y^{\prime}_1,y^{\prime}_2\}\\
\frac{\delta_2}{2(R+D)}\pi(C_y)\ \ &\text{if}\ i\in (V(\P^{\prime}_{yv}) \cup V(\P^{\prime}_{vy})) \backslash \{v^{\prime}\}\\
\lf(1-\delta_2(1-\frac{\delta_2}{2(R+D)})\rf)\pi_i\ \ &\text{if}\ i\in V(G)\\
\frac{\delta_2}{2(R+D)}\ \ &\text{if}\ i=v^{\prime}
\end{cases}
\end{align*}}
To guarantee that this chain is indeed the pseudo-lifting of the original random walk
$P$, consider $T = V(G)$ and $\delta_2$, where
\begin{equation}
\frac12 = \delta_2 \lf(1-\frac{1}{2(R+D)}\rf).\label{eq:delta2}
\end{equation} Note that $\widehat{G}$ has exactly
$|E|+2Rn+2D|Y|$ edges.
\vspace{-.1in}

\subsection{Mixing time and Size: Proof of Theorem \ref{thm:performance_pseudo_lifting}}\label{subsec:four-two}
\vspace{-.1in}

We prove two Lemmas about the performance of pseudo-lifting we constructed, and they imply
Theorem \ref{thm:performance_pseudo_lifting}.
At first, we state the
following result about its mixing time, and the proof
can be done similarly as the proof of Theorem \ref{thm:mixing_pseudo}.
\begin{lemma}\label{lemma:pseudo_mixing}
The mixing time of the random walk $\widehat{P}$ defined by $\widehat{Q}$ is
$O(D)$.
\end{lemma}
\begin{proof}
Consider the following stopping rule. Walk until visiting
$v^{\prime}$, and toss a coin $X$ with the following probability.
{\small
\begin{align*}
X=\begin{cases} 0\ &\text{with probability}\
\frac{\delta_2}{2(R+D)}\\
1\ \ &\text{with probability}\
\frac{\delta_2 D}{2(R+D)}\\
2\ \ &\text{with probability}\
\frac{\delta_2(R-1)}{2(R+D)}\\
3\ \ &\text{with probability}\
1-\delta_2(1-\frac{\delta_2}{2(R+D)})\\
4\ \ &\text{with probability}\
\frac{\delta_2(R-1)}{2(R+D)}\\
5\ \ &\text{with probability}\
\frac{\delta_2 D}{2(R+D)}\\
\end{cases}
\end{align*}}
Depending on the value of $X$,
\begin{itemize}
\item[$\circ$] $X=0$ : Stop at $v^{\prime}$.
\item[$\circ$] $X=1$ : Walk on a directed path $\P^{\prime}_{vy}$, and choose
its interior node uniformly at random, and stop there.
\item[$\circ$] $X=2$ : Walk until getting a directed path $\P^{\prime}_{yw}$, and choose
its interior node uniformly at random, and stop there.
\item[$\circ$] $X=3$ : Walk until getting an old node in $V(G)$, and stop there.
\item[$\circ$] $X=4$ : Walk until getting a directed path $\P^{\prime}_{wy}$, and choose
its interior node uniformly at random, and stop there.
\item[$\circ$] $X=5$ : Walk until getting a directed path $\P^{\prime}_{yv}$, and choose
its interior node uniformly at random, and stop there.
\end{itemize}
It can be checked, using arguments similar to that in proof of Theorem
\ref{thm:mixing_pseudo}, that the distribution of the stopped node is
precisely $\widehat{\pi}$. Also,
we can show that the
expected length of this stopping rule is
$O(\frac{R+D}\delta_2)=O(\frac{D}\delta_2)=O(D)$ from (\ref{eq:delta2}).
This is primarily true because the probability of getting on a directed
path $\P^{\prime}_{wy}$ at $w$ is $\Theta(\delta_2/(R+D))$.
\end{proof}

Now we apply the hierarchical construction to the case of graphs with constant doubling dimension,
and show the guarantee for the size of the pseudo-lifting in terms of its doubling dimension.

\begin{lemma}\label{lemma:pseudo_size}
Given a graph $G$ with a constant doubling dimension $\rho$ and its diameter
$D$, the hierarchical construction gives
a pseudo-lifted graph $\hat{G}$ with its size $|\widehat{E}|=O(Dn^{1-\frac1{\rho+1}})$.
\end{lemma}
\begin{proof}
The property of doubling dimension graph implies that there exists an
$R$-net $Y$ such that $|Y|\leq (2D/R)^\rho$ (cf. \cite{A83}). Consider
$R=D2^\frac{\rho}{\rho+1}n^{-\frac1{\rho+1}}$. This is an appropriate choice
because
$R=D2^\frac{\rho}{\rho+1}n^{-\frac1{\rho+1}}>Dn^{-\frac1{\rho+1}}>n^{\frac1{\rho}-\frac1{\rho+1}}>1$
(the second inequality is from $n \leq D^\rho$).
Given this, the size of the pseudo-lifted graph $\hat{G}$ is
$$|\widehat{E}|=|E|+2Rn+2D|Y|\leq|E|+2D\lf(\frac{2^\frac{\rho}{\rho+1}}{n^{\frac1{\rho+1}}}\rf)n+
2D\lf(2\frac{n^{\frac1{\rho+1}}}{2^\frac{\rho}{\rho+1}}\rf)^\rho=|E|+O(Dn^{1-\frac1{\rho+1}}).$$
Since $|E| = O(n)$ and $D = \Omega(n^{1/\rho})$, we have that $|\hat{E}| = O(Dn^{1-\frac1{\rho+1}})$.
\end{proof}

\vspace{-.1in}
\section{Application: Back to Averaging}\label{sec:five}
\vspace{-.1in}

As we introduced in the introduction,
consider the following computation problem of the distributed averaging.
Given a connected network graph $G=(V,E)$, where $V=\{1,2,\ldots n\}$, each node $i\in V$ has a value $x_i \in \mathbb{R}$.
Then the goal is to compute the average of $\bx=[x_i]$ only by communications between adjacent nodes:
\beq
x_{ave}=\frac1n \sum_i x_i.
\eeq
This problem arises in many applications such
as distributed estimation \cite{T84}, distributed spectral decomposition \cite{KM04}, estimation and
distributed data fusion on ad-hoc networks \cite{MFHH02}, distributed sub-gradient method for eigenvalue maximization \cite{bgps}, inference in Gaussian graphical models \cite{MJW06}, and coordination of autonomous agents \cite{JLS03}.

\vspace{-.1in}
\subsection{Linear iterative algorithm}\label{subsec:five-one}
\vspace{-.1in}

A popular and quite simple approach for this computation is a method based on linear iterations \cite{BX04} as follows.
Suppose we are given with a graph conformant random walk $P$ which has the uniform stationary distribution $\pi$
i.e. $\pi^T P=\pi^T$.
The linear iteration algorithm is described as follows. At time $t$, each node $i\in V$ has an estimate $y_i(t)$ of $x_{ave}$ and initially $y_i(0)=x_i$. At time $t=1,2,\ldots$ for each edge $(i,j)$ of $G$, node $i$ sends value $P_{ji} y_i(t)$ to node $j$. Then each node $j$ sums up the values received as its estimate at time $t+1$, that is
$$y_j(t+1)=\sum_{i=1}^n P_{ji} y_i(t).$$
Under the condition that $P$ is ergodic, i.e. $P$ is connected and aperiodic, it is known that \cite{BX04}
$$\lim_{t\rightarrow \infty} y(t)=\lim_{t\rightarrow \infty} P^t \bx = \left(\sum_i x_i\right) \pi =\frac1n\sum_i x_i\mathbf{1}
=x_{ave}\mathbf{1},~ \mbox{where} ~ \mathbf{1}=[1].$$
Specifically, as we already saw in the introduction, $\ep$-computation time $T_\ep(P)$ is defined as:
\beq
T_\ep(P)& = & \inf\left\{t~:~\forall \bx \in \R_+^n,
\frac{\|P^t\bx - x_{ave} \bone\|_\infty}{x_{ave}} \le \ep \right\}.
\eeq
The quantity $T_\ep(P)$ is well known to be related to the mixing time $\mathcal{H}(P)$. More precisely,
we prove Lemma \ref{lemma:running_mixing}, which implies \beq T_\ep(P)=O^*\left(\mathcal{H}(P) \log \frac{1}{\ep}  \right).\label{eq:mixing}\eeq
Since each edge $(i,j)$ such that $P_{ij}>0$ performs an exchange of values per each iteration,
the number of operations performed per iteration across
the network is at most $|E|$.
Thus, the total number of operations of the linear iterations to obtain the approximation of $x_{ave}$ scales like
\beq C_\ep(P):=T_\ep(P)\times |E|.\eeq
Therefore, the task of designing an appropriate $P$ with small $\mathcal{H}(P)$
is important to minimize both $T_\ep(P)$ and $C_\ep(P)$.

\vspace{-.1in}
\subsection{Linear iterative algorithm with pseudo-lifting: Proof of Theorem \ref{thm:performance_averaging}}\label{subsec:five-two}
\vspace{-.1in}

We present a linear iterative algorithm that utilizes the pseudo-lifted version
of a given matrix $P$ on the original graph $G$. The main idea behind this
implementation is to run the standard linear iterations in $\widehat{G} = (\hat{V},\hat{E})$
with the pseudo-lifted chain $\widehat{P}$. However, we wish to implement
this on $G = (V, E)$ and not $\widehat{G}$. Now recall that $\widehat{G}$
has the following property: (a) each node $\hat{v} \in \hat{V}$ is a copy of
a node $v \in V$, and (b) each edge $(\hat{u}, \hat{v})$ is a copy of
edge $(u,v) \in E$, where $\hat{u}, \hat{v}$ are copies of $u, v \in V$
respectively. Therefore, each node $\hat{v} \in V$ can be {\em simulated}
by a node $v \in V$ where $\hat{v}$ is a copy of $v$ for the purpose of linear
iterations. Thus, it is indeed possible to simulate the pseudo-lifted version of
a matrix $P$ on $G$ by running multiple threads (in the language of the computer
programming) on each node of $G$. We state this approach formally as follows:

\begin{itemize}
\item[1.] Given graph $G=(V,E)$, we wish to compute the average $x_{ave}$
at all nodes. For this, first produce a matrix $P$ using the Metropolis-Hastings
method with the uniform stationary distribution.

\item[2.] Construct the pseudo-lifting $\widehat{P}$ based on $P$ as explained in
 Section \ref{sec:four}. This pseudo-lifted random walk has a stationary distribution
 $\widehat{\pi}$ on a graph $\widehat{G}$.

\item[3.] As explained below, implement the linear iterative algorithm based on
$\widehat{P}$ on the original graph $G$.

\begin{itemize}
\item[$\circ$] Let $t$ be the index of iterations of the algorithm and initially
it be equal to $0$.

\item[$\circ$] For each node $\hat{v} \in \hat{V}$, maintain a number
$y_{\hat{v}}(t)$ at the $t^{th}$ iteration. This is maintained at the node $v \in V$
where $\hat{v}$ is a copy of $v$. The initialization of these values is
stated below.
\begin{itemize}

\item[$\bullet$] Recall that, $\hat{V}$ contains $V$
as its subset. Recall that they are denoted as $V(G) \subset \hat{V}$, and
each $v\in G$ has its copy $\bar{v}\in V(G)$.

\item[$\bullet$] For each $\hat{v} \in V(G)$, initialize $y_{\hat{v}}(0)= x_v$.

\item[$\bullet$] For each $\hat{v} \in V \backslash V(G)$, initialize $y_{\hat{v}}(0) = 0$.
\end{itemize}

\item[$\circ$] In the $t+1^{th}$ iteration, update
$$ y_{\hat{v}}(t+1) = \sum_{\hat{u} \in \hat{V}} \widehat{P}_{\hat{v}\hat{u}} y_{\hat{u}}(t).$$
This update is performed by each node $v$ through receiving information from
its neighbors $u$ in $G$, where $\hat{v}$ is a copy of $v$ and neighbors (of $\hat{v}$)
$\hat{u}$ are copies of neighbors (of $v$) $u$.
\end{itemize}

\item[4.] At the end of the $t^{th}$ iteration, each node $v$ produces its estimate as $2y_{\hat{v}}(t)$, $\hat{v} \in V(G)$.
\end{itemize}
It can be easily verified that since above algorithm is indeed implementing the linear
iterative algorithm based on $\hat{P}$, the $\varepsilon$ computation time is
$T_{\varepsilon}(\widehat{P})$ and the total number of communications performed
is $C_\ep(\widehat{P})$.
In what follows, for the completeness we bound $T_{\varepsilon}(\widehat{P})$ and $C_\ep(\widehat{P})$.
\begin{lemma}\label{lemma:running_mixing}
$T_{\varepsilon}(\widehat{P})=O\lf(\mathcal{H}(\widehat{P})\log \frac1{\varepsilon\pi_0}\rf)$.
\end{lemma}

\begin{proof}
Here, we need the $\varepsilon$-mixing time $\tau(\varepsilon)$ based on the total variance distance,
and recall its definition in Section \ref{subsec:two-four}:
$$\tau(\beps)=\min\lf\{t:\forall i\in G, \frac{1}{2} {\sum_{j\in G}
\left|P^t_{ij}-\pi_j\right|}\leq \beps \rf\}.$$
The following relation between two different mixing time $\tau(\varepsilon)$ and $\mathcal{H}$
is known (see \cite{LW98}):
$$\tau(\varepsilon) =
O\lf(\mathcal{H}\log{\frac1{\varepsilon}}\rf).$$
If $t$ is larger than $\tau(\varepsilon\pi_0/4)$ of $\widehat{P}$ , which is
$O\lf(\mathcal{H}(\widehat{P})\log \frac1{\varepsilon\pi_0/4}\rf)$,
\begin{eqnarray*}
\lf| y_i(t)-\langle y(0),\widehat{\pi}\rangle\rf| &= &  \lf|\sum_j \widehat{P}^t_{ij}y_j(0)-\sum_jy_j(0)
\widehat{\pi}_j\rf|~ \leq~ \sum_j y_j(0)\lf|\widehat{P}^t_{ij}-\widehat{\pi}_j\rf|\nonumber\\
&\stackrel{(a)}{\leq} &  \sum_j y_j(0)\frac{\varepsilon \pi_0}2
~ \stackrel{(b)}{\leq} ~ \sum_j y_j(0)\varepsilon \widehat{\pi}_j
~= ~\varepsilon \langle y(0),\widehat{\pi}\rangle,\nonumber
\end{eqnarray*}
where (a) is from
$\lf|\widehat{P}^t_{ij}-\widehat{\pi}_j\rf|\leq
\sum_j \lf|\widehat{P}^t_{ij}-\widehat{\pi}_j\rf|\leq 2\times\frac{\varepsilon\pi_0}4=\frac{\varepsilon \pi_0}2$,
and (b) is because $\widehat{\pi}_j>\frac12\pi_j\geq\frac12\pi_0$ for every old node $j\in V(G)$,
and $y_j(0)=0$ otherwise. This completes the proof.
\end{proof}
From the proof of Lemma \ref{lemma:running_mixing},
note that the relation $T_{\varepsilon}(P)=O\lf(\mathcal{H}(P)\log \frac1{\varepsilon\pi_0}\rf)$
holds for any random walk $P$. Therefore,
$T_{\varepsilon}(\widehat{P})=O\lf(D \log \frac1{\varepsilon\pi_0}\rf)$ and
$C_\ep(\widehat{P})=T_{\varepsilon}(\widehat{P})\times |\widehat{E}|=
O\lf(D^2n^{1-\frac{1}{1+\rho}}
\log \frac1{\varepsilon\pi_0}\rf)$
since $\mathcal{H}(\widehat{P})=O(D)$ and $|\widehat{E}|=O(Dn^{1-\frac{1}{1+\rho}})$
from Lemma \ref{lemma:pseudo_mixing} and \ref{lemma:pseudo_size}.
This also completes the proof of Theorem \ref{thm:performance_averaging}.

\vspace{-.1in}
\subsection{Comparison with other algorithms}\label{subsec:five-three}
\vspace{-.1in}

Even considering any possible algorithms based on passing messages,
the lower bound of the performance guarantees in the averaging problem
is $O(D)$ for the running time, and $O(Dn)$ for the total number of operations.
Therefore, our algorithm using pseudo-lifting gives the best running time, and
possibly loses $\frac{O^*(D^2n^{1-\frac1{\rho+1}})}{O(Dn)}=O^*(D/n^{\frac1{\rho+1}})$ factor
in terms of the total number of operations compared to the best algorithm.
For example, when $G$ is a $d$-dimensional grid graph, this loss is only $O^*(D/n^{\frac1{\rho+1}})=O^*(n^{1/d}/n^{\frac1{d+1}})
=O^*(n^{\frac1{d(d+1)}})$ since the doubling dimension of $G$ is $d$ and its diameter $D$ is $O(n^{1/d})$.
The standard linear iterations using the Metropolis-Hastings method
loses $\Omega(n^{1/d})$ factor in both the running time and the total number of operations (see Table I).

\begin{table}
\begin{center}
\begin{tabular}{|c|c|c|c|}
\hline
&Metropolis-Hastings&Pseudo-Lifting& Optimal\\
\hline\hline
$\mathop{\mbox{Mixing time(Running time)}}\limits_{:~d\mbox{-dim. grid graph}}$
& $\mathop{\Omega\left(\frac1{\Phi^2(P)}\right)}\limits_{:~O^*\left(n^{\frac2d}\right)}$
& $\mathop{O(D)}\limits_{:~O\left(n^{\frac1d}\right)}$
& $\mathop{D}\limits_{:~n^{\frac1d}}$\\
\hline
$ \mathop{\mbox{Size (dbl. dim. $\rho$)}}\limits_{:~d\mbox{-dim. grid graph}}$
& $\mathop{\Theta(n)}\limits_{:\Theta(n)}$
& $\mathop{O\left(n^{\frac{\rho}{\rho+1}} D\right)}\limits_{:~O\left(n^{1+\frac{1}{d(d+1)}}\right)}$
& $\mathop{n}\limits_{:~n}$\\
\hline
$\mathop{\mbox{Total }\#\mbox{ of operations}}\limits_{:~d\mbox{-dim. grid graph}}$
& $\mathop{\Omega\left(\frac{n}{\Phi^2(P)}\right)}\limits_{:~O^*\left(n^{1+\frac2d}\right)}$
& $\mathop{O\left(n^{\frac{\rho}{\rho+1}} D^2\right)}\limits_{:~O^*\left(n^{1+\frac{d+2}{d(d+1)}}\right)}$
& $\mathop{nD}\limits_{:~n^{1+\frac1d}}$\\
\hline

\end{tabular}
\end{center}
\label{tablex} \caption{Comparison of pseudo-lifting with the Metropolis-Hastings method.
Here, we assume $G$ has $\Theta(n)$ edges.}
\end{table}
We take note of the following subtle matter: the non-reversibility is captured in
the transition probabilities of the underlying Markov chain (or random walk); but
the linear iterative algorithm does not change its form other than this detail.

\vspace{-.1in}
\section{Lifting Using Expanders}\label{sec:six}
\vspace{-.1in}

We introduced the new notion of pseudo-lifting for the applications of interest,
one of which was the distributed averaging.
However, since it may not be relevant to certain other applications,
we optimize the size of lifting (not pseudo-lifting) in \cite{CLP02}.
The basic motivation of our construction is using the expander graph, instead of the complete graph
in \cite{CLP02}, to reduce the size of the lifting.

\vspace{-.1in}
\subsection{Preliminaries}
\vspace{-.1in}

In what follows, we will consider only $P$ such that $P\geq I/2$. This
is without loss of generality due to the following reason. Suppose
such is not the case,  then we can modify
it as $(I+P)/2$; the mixing time of $(I+P)/2$ is within a constant factor
of the mixing time of $P$.
\vspace{-.1in}
\subsubsection{Multi-commodity Flows}\label{subsec:two-three}
\vspace{-.1in}
In \cite{CLP02}, the authors use a multi-commodity flow to construct a
specific lifting of a given random walk $P$ to speed-up its mixing
time. Specifically, they consider a multi-commodity flow problem on
$G$ with the capacity constraint on edge $(u, v) \in E$ given by
$Q_{uv}$. A flow from a source $s$ to a destination $t$, denoted by $f$, is
defined as a non-negative function on edges of $G$ so that
$$\sum_j f(ji)=\sum_j f(ij)$$ for every node $i\neq s,t$. The value of the flow is defined by
$$val(f) = \sum_{j} f(sj) -\sum_{j} f(js) = \sum_{j} f(jt)-\sum_{j} f(tj)$$, and  the cost
of flow $f^{st}$ is defined as $$cost(f) = \sum_{(i,j)\in E}
f(ij).$$ A {\em multi-commodity flow} is a collection $f=(f^{st})$ of flows, where each $f^{st}$ is a flow from $s$ to $t$. Define
the {\em congestion} of a multi-commodity flow $f$ as $$\max_{(i,j)\in E} \frac{\sum_{s,t} f^{st}(ij)}{Q_{ij}}.$$
Consider the following optimization problem, essentially trying to minimize
the congestion and the cost simultaneously under the condition for the amount of flows:
\begin{align*}
&{\sf minimize}~~~~  K\\
&{\sf subject~ to} ~~   val(f^{st})=\pi_s\pi_t, ~~\forall s,t,\\
&~~~~~~~~~~~~~~ \sum_{s, t} f^{st}(ij) \leq K Q_{ij}, ~~\forall (i,j) \in E, \\
&~~~~~~~~~~~~~~ \sum_t  cost(f^{st}) \leq K \pi_s, ~ \sum_s cost(f^{st}) \leq K \pi_t, ~\forall s, t.
\end{align*}
Let $C$ be the optimal solution of the above problem. It is easy to
see that $C\geq 1/\Phi$. Further, if $P$ is reversible, then result
of Leighton and Rao \cite{LR88} on the approximate multi-commodity
implies that
$$C =  O\left(\frac1{\Phi} \log\frac{1}{\pi_0}\right).$$
Let the optimal multi-commodity flow of the above problem be $F_1$,
and we can think of $F_1$ as a weighted collection of directed
paths. In \cite{CLP02}, the authors modified $F_1$, and got a new
multi-commodity flow $F_2$ that has the same amount of $s-t$ flows as
$F_1$, while its congestion and path length are at most $12C$. They used
$F_2$ to construct a lifting $\hat{P}$ with mixing time
$\widehat{\cH}$ such that
$$\widehat{\cH} \leq 144C.$$
Also, they showed that the mixing time of any lifting $\hat{P}$ is greater than $C/2$,
hence their lifted Markov chain has almost optimal speed-up within a constant factor.

To obtain a lifting with the smaller size than that in \cite{CLP02}, we will to study the existence
of the specific $k$-commodity flow with short path lengths. For this, we
will use a {\em balanced multi-commodity flow},
which is a multi-commodity flow with the following
condition for the amount of flows:
$$val(f^{st})=g(s,t),\forall s,t,$$
and $g(s,t)$ satisfies the balanced condition:
$$ \sum_t g(s,t) \leq \pi_s, ~\sum_s g(s,t) \leq \pi_t, ~~\forall s,t.$$
Therefore, $F_1$ and $F_2$ are also balanced multi-commodity flows with $g(s,t)=\pi_s\pi_t$.
Given a multi-commodity flow $f$, let $C(f)$ be its congestion
and $D(f)$ be the length of the longest flow-path. Then, the {\em flow
number} $T$ is defined follows:
$$T=\min_{f} \lf(\max\lf\{C(f),D(f)\rf\}\rf),$$
where the minimum is taken over all balanced multi-commodity flows with
$g(s,t)=\pi_s\pi_t$.
Hence, $F_2$ implies $T \leq 12C$.
The following claim appears in \cite{KS02}:
\begin{claim}\label{clm:fn}(Claim 2.2 in \cite{KS02})
For any $g(s,t)$ satisfying the balanced condition (not necessarily $g(s,t)=\pi_s\pi_t$),
there exists a balanced multi-commodity flow $f$ with $g(s,t)$
such that $\max\{C(f),D(f)\}\leq 2T$.
\end{claim}

\vspace{-.1in}
\subsubsection{Expanders}
\vspace{-.1in}
The expander graphs are sparse graphs which have high
connectivity properties, quantified using the edge expansion $h(G)$ as
defined as
$$h(G)=\min_{1\leq |S| \leq \frac{n}2} \frac{|\partial(S)|}{|S|},$$
where $\partial(S)$ is the set of edges with exactly one endpoint in
$S$. For constants $d$ and $c$, a family
$\mathcal{G}=\{G_1,G_2,\dots\}$ of $d$-regular graphs is called a
$(d,c)$-expander family if $h(G)>c$ for every $G\in \mathcal{G}$.
There are many explicit constructions of a $(d,c)$-expander family
available in recent times. We will use a $(d,c)$-expander graph
$G^{Ex}=(V,E^{Ex})$ (i.e. $V^{Ex}=V$), and a transition matrix
$P^{Ex}$ defined on this graph. For a given $\pi$, we can define a
reversible $P^{Ex}$ so that its stationary distribution is $\pi$ as
follows,
$$P^{Ex}_{ij} =
\begin{cases}
    \frac{\pi_0}{d\pi_i} & \text{if }(i,j)\in E^{Ex}\\
    1-\frac{\pi_0}{\pi_i} & \text{if }i= j\\
  \end{cases}.$$
In the case of $\pi_{max}=O(\pi_0)$, it is easy to check that
$\Phi(P^{Ex})=\Theta(h(G))=\Omega(1)$, where $\Phi(P^{Ex})$ is the
conductance of $P^{Ex}$. Hence,
$\lambda_{P^{Ex}} = \Omega(1)$, and the random walk defined by
$P^{Ex}$ mixes fast. In this Section, we will consider
only such $\pi$.
\vspace{-.1in}
\subsection{Construction}\label{subsec:four-one}
\vspace{-.1in}
We use the multi-commodity
flow based construction which was introduced in \cite{CLP02}.  They  essentially
use a multi-commodity flow between source-destination pairs for all
$s, t \in V$. Instead, we will use a balanced multi-commodity flow between
source-destination pairs that are obtained from an expander. Thus,
the essential change in our construction is the use of an expander in
place of a complete graph used in \cite{CLP02}. A caricature of this lifting is explained in
Figure \ref{fig5}. However, this
change makes the analysis of the mixing time lot more challenging and
requires us to use different analysis techniques. Further,
we use arguments based on the classical linear programming
to derive the bound on the size of lifting.

\begin{figure}[htb]
\begin{center}
\centerline{\psfig{figure=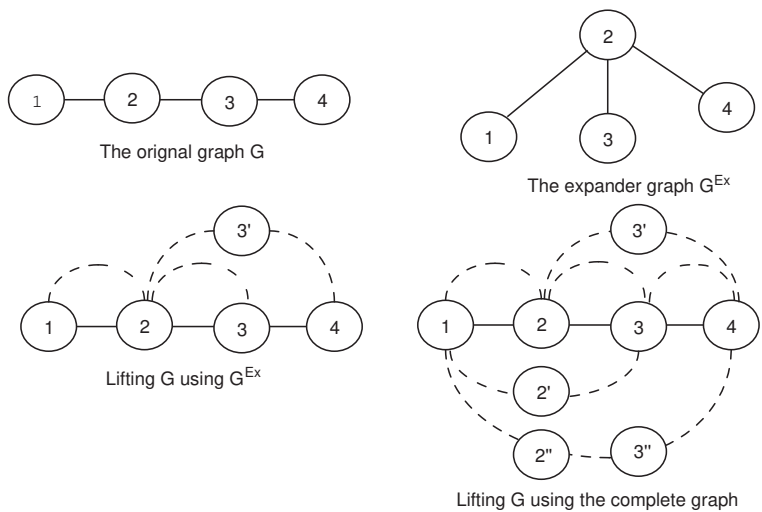,width=10cm,angle=0}}
\caption{A caricature of lifting using expander.
Let line graph $G$ be a line graph with $4$ nodes.
We wish to use an expander $G^{Ex}$ with 4 nodes, shown on the top-right side of the figure.
$G$ is lifted by adding paths that correspond to edges of expander. For example, an edge
$(2,4)$ of expander is added as path $(2, 3', 4)$.
We also draw the lifting in \cite{CLP02} which uses the complete graph.}
\label{fig5}
\end{center}
\end{figure}

To this end, we consider the following multi-commodity flow: let
$G^{Ex}=(V, E^{Ex})$ be an expander with a transition matrix
$P^{Ex}$ and a stationary distribution $\pi$ as required -- this is feasible
since we have assumed $\pi_{max} = O(\pi_0)$. We note that this assumption
is used only for the existence of expanders. Consider
a multi-commodity flow $f = (f^{st})_{(s,t) \in E^{Ex}}$ so that
\begin{itemize}
\item[(a)] $val(f^{st})=\pi_sP^{Ex}_{st}=Q^{Ex}_{st},\forall (s, t) \in E^{Ex}$;
\item[(b)]  $\sum_{s, t} f^{st}(ij) \leq K Q_{ij},\forall (i,j) \in E$;
\end{itemize}
\begin{lemma}\label{lemma:bmfp}
There is a feasible multi-commodity flow in the above flow problem with congestion($K$) and path-length at most $W$, where $W=O^*(1/\Phi(P))$.
\end{lemma}
\begin{proof}
The conclusion
is derived directly from Claim \ref{clm:fn} since the
flow number $T$ is less than $12C=O^*(1/\Phi(P))$ and the flow considered
is a balanced multi-commodity flow i.e. $W=24C=O^*(1/\Phi(P))$.
\end{proof}
Now, we can
think of this multi-commodity flow as a weighted collection of directed paths
$\{(\P_r,w_r):1\leq r\leq N\}$, where the total weight of paths from
node $s$ to $t$ is $\pi_s P^{Ex}_{st}$, where $(s,t) \in E^{Ex}$.
Let $\ell_r$ be the length of path $\P_r$. From Lemma \ref{lemma:bmfp},
we have the following:
\begin{equation}
\sum_r w_r=1,\ \ \ \ \ell_r\leq W,
\end{equation}
\begin{equation}
 \sum_{r:\P_r\text{ starts at
}i  }w_r=\pi_i,\ \sum_{r:\P_r\text{ ends at }i}w_r=\pi_i, ~\mbox{for $i \in V$}
\end{equation}
\begin{equation}
\sum_{r:(i,j) \in E(\P_r)}w_r\leq W Q_{ij}, ~\mbox{for $(i,j) \in E$.}
\end{equation}

Using such a collection of weighted paths, we construct the
desired lifting next. As Figure \ref{fig5}, we construct the lifted graph
$\widehat{G}=(\widehat{V},\widehat{E})$ from $G$ by adding a
directed path $\P^{\prime}_r$ of length $\ell_r$ connecting $i$ to $j$
if $\P_r$ goes from $i$ to $j$. Subsequently, $\ell_r-1$ new nodes
are added to the original graph. The ergodic flow on an edge
$(i, j)$ of the lifted chain is defined by
\begin{align*}
\widehat{Q}_{ij}=\begin{cases}
w_r/2W\ \ &\text{if}\ (i, j) \in E(\P^{\prime}_r)\\
Q_{ij}-\sum_{r:ij\in E(\P_r)} w_r/2W &\text{if}\ (i, j) \in E(G)
\end{cases}
\end{align*}
It is easy to check it defines a Markov chain on $\widehat{G}$, and
a natural way of mapping the paths $\P^{\prime}_r$ onto the paths
$\P_r$ collapses the random walk on $\widehat{G}$ onto the random
walk on $G$. The stationary distribution of the lifted chain is
\begin{align*}
\widehat{\pi}_i=\begin{cases}
w_r/2W\ \ &\text{if}\ i\in V({\P^{\prime}}_r)\backslash V(G)\\
\pi_i-\sum_{r:\P_r\text{ thru }i} w_r/2W &\text{if}\ i\in
V(G)
\end{cases}
\end{align*}
Thus, the above stated construction is a valid lifting of the given Markov chain $P$
defined on $G$.
\vspace{-.1in}
\subsection{Mixing time and size: Proof of Theorem \ref{thm:performance_expander_lifting}}\label{subsec:four-two}
\vspace{-.1in}
We prove two Lemmas about the performance of lifting we constructed, and they imply
Theorem \ref{thm:performance_expander_lifting}.
At first, we state and prove the lemma which bounds the mixing time of the lifted chain we constructed.
\begin{lemma}\label{lemma:mixing_expander}
The mixing time $\widehat{\mathcal{H}}$ of the lifted Markov chain
represented by $\hat{Q}$ defined on $\hat{G}$ is $O^*(1/\Phi(P))$
\footnote{The precise bound is $O(W\log \frac1{\pi_0})$.}.
\end{lemma}
\begin{proof}
By the property of expanders, we have $\lambda_{P^{Ex}} = \Omega(1)$. Therefore,
it is sufficient to show that
$$ \hat{\cH} = O\lf(\frac{W}{\lambda_{P^{Ex}}}\log
\frac1{\pi_0}\rf).$$
First, note that for any node $i \in V$ (i.e. a original node $i$ in $G$),
\beq
\frac12 \pi_i \leq \widehat{\pi}_i \leq \pi_i.
\label{efd1}
\eeq
Now, under the lifted Markov chain the probability of
getting on any directed path $\P^{\prime}_r$ starting at
$i$ is
$$\widehat{P}_{ij}=\frac{\widehat{Q}_{ij}}{\widehat{\pi}_i}=\frac{w_r}{2W\widehat{\pi}_i}.$$
Hence the probability of getting on any directed path starting at
$i$ is
$$\sum_{r:\P^{\prime}_r\text{ starts at
}i}\frac{w_r}{2W\widehat{\pi}_i}=\frac1{2W\widehat{\pi}_i}
\sum_{r:\P^{\prime}_r\text{ starts at }i}
w_r=\frac{\pi_i}{2W\widehat{\pi}_i}.$$ From (\ref{efd1}),
this is bounded between $\frac1{2W},$ and
$\frac1{W}$.

To study the $\hat{\cH}$, we will focus on the induced random walk (or Markov chain)
on original nodes $V \subset \hat{V}$ by the lifted Markov chain $\hat{P}$.
Let $\widehat{P}^V$ be the transition matrix of this induced random walk.
Then,
$$\widehat{P}^V_{ij}=\widehat{P}_{ij}+\sum_{r:\P^{\prime}_r\text{
goes from }i\text{ to
}j}\frac{w_r}{2W{\widehat{\pi}_i}}.$$
Now, $\widehat{P}^V\geq \widehat{P} \geq I/4$,
because
$\widehat{P}_{ii}=\widehat{Q}_{ii}/\widehat{\pi}_i\geq
Q_{ii}/2\widehat{\pi}_i=P_{ii}\pi_i/2\widehat{\pi}_i\geq
P_{ii}/2\geq I/4.$ Here we have assumed that $P \geq I/2$
as discussed earlier. Now,
\begin{align*}
\widehat{P}^V_{ij}&\geq
\frac{1}{2W{\widehat{\pi}_i}}\sum_{r:\P^{\prime}_r\text{
goes from }i\text{ to
}j}w_r=\frac{\pi_iP^{Ex}_{ij}}{2W{\widehat{\pi}_i}}\geq
\frac1{2W}P^{Ex}_{ij}.
\end{align*}
And, its stationary distribution $\widehat{\pi}^V$ is  :
$\widehat{\pi}^V_i=\frac{\widehat{\pi}_i}{\widehat{\pi}(V)}.$
Therefore, by (\ref{efd1}) we have
$\frac12 \pi_i \leq \widehat{\pi}^V_i \leq 2 \pi_i.$
Now, we can apply Claim \ref{clm:a2} to obtain
the following:
\beq
\lambda_{\widehat{P}^V(\widehat{P}^V)^*}
=\Omega\lf(\frac1{W}\lambda_{P^{Ex}}\rf).\label{efd2}
\eeq
Now, we are ready to design the following stopping rule $\Gamma$ that will imply that
the desired bound on $\hat{\cH}$.
\begin{itemize}
\item[(i)] Walk until visiting old nodes of $V \subset \hat{V}$ for $T$ times, where
$ T := \lf\lceil2\log (2/\widehat{\pi}^V_0)/
{\lambda_{\widehat{P}^V(\widehat{P}^V)^*}} \rf\rceil. $ Let this
$T^{th}$ old node be denoted by $X$.
\item[(ii)] Stop at $X$ with probability $1/2$.
\item[(iii)] Otherwise, continue walking until getting onto any directed path
$\P^{\prime}_r$; choose an interior node $Y$ of $\P^{\prime}_r$
uniformly at random and stop at $Y$.
\end{itemize}
From the
relation (\ref{rel:mixing_eigenvalue}) in Section \ref{subsec:two-four}
with $\varepsilon=\frac12\sqrt{\widehat{\pi}^V_0}$,
it follows that after
time $T$ as defined above the Markov chain $\widehat{P}^V$, restricted to old nodes $V$,
has distribution close to $\hat{\pi}^V$ i.e.
$$ |\Pr(X=w)-\widehat{\pi}^V_w| \leq \widehat{\pi}^V_w/2, ~~\forall ~w \in V.$$
According to the above stopping rule, we stop at an old node $w$
with probability $1/2$. Therefore, for any $w \in V$, we have that
the stopping time $\Gamma$ stops at $w$ with probability at least
$\widehat{\pi}^V_w/4\geq{\pi}_w/8\geq \widehat{\pi}_w/8$. With
probability $1/2$, the rule does not stop at the node $X$. Let $w^k$
be the $k^{th}$ point in the walk starting from $X$. Because at any old
node $i$, the probability of getting on any directed path is between
$\frac1{2W}$ and $\frac1{W}$, a coupling
argument shows that for any old node $i$,
$$\Pr(w^k=i|w^0,\cdot,w^k\text{ are old
nodes})\geq
\left(1-\frac1{W}\right)^k\frac12\widehat{\pi}^V_i$$ If
$w$ is a new point on the directed path ${\P^{\prime}}_r$ which
connects the old node $i$ to $j$. Then,
\begin{align*}
\Pr(\Gamma \text{~stop at }
w)&\geq\frac12\sum_{k=0}^{\infty}\text{Prob}(w^k=i|w^0,\cdot,w^k\text{
are old points})\\
&~~~~~~~~~~~~~~~~\times\text{Prob}(\text{at }i\text{, get on the
path~}{\P^{\prime}}_r)\times\frac1{\ell_r}\\
&\geq\frac12\sum_{k=0}^{\infty}\left(1-\frac1{W}\right)^k\frac12\widehat{\pi}^V_i\frac{w_r}{2W\widehat{\pi}_i}\frac1{W} \\&\geq \frac{w_r}{16W^2}\sum_{k=0}^{\infty}\left(1-\frac1{W}\right)^k\\
&=\frac{w_r}{16W}\\
&=\frac18\widehat{\pi}_w
\end{align*}
The average length of this stopping rule is $O(T+W)$. By
(\ref{efd2}), {\small \begin{align*}
O(T+W)&=O\left(\lf\lceil\frac2{\lambda_{\widehat{P}^V(\widehat{P}^V)^*}}\log
(2/\pi_0)\rf \rceil+W \right)=
O\left(\frac{W}{\lambda_{P^{Ex}}}\log (1/\pi_0)\right).
\end{align*}}
Thus, we have established that the stopping rule $\Gamma$ has the average
length $O(W \log 1/\pi_0)$ and the distribution of the stopping node is
$\Omega(\hat{\pi})$. Therefore, using the {\em fill-up} lemma stated in \cite{A82}, it follows that
$\widehat{\cH} = O(W\log 1/\pi_0)$.
\end{proof}

Also, we bound the size of the lifted chain we constructed as follows.
\begin{lemma}\label{lemma:lift-size}
The size of the lifted Markov chain can be bounded above as $O^*(|E|/\Phi(P))$
\footnote{The precise bound is $O(|E|W)$.}.
\end{lemma}
\begin{proof}
We want to establish that the size of the lifted chain in terms of
the number of edges, i.e.  $|\hat{E}| = O^*(|E|/\Phi(P))$. Note that, the lifted
graph $\hat{G}$ is obtained by adding paths that appeared in
the solution of the multi-commodity flow problem. Therefore, to
establish the desired bound we need to establish a bound on the
number of distinct paths as well as their lengths.

To this end, let us re-formulate the multi-commodity flow based on
expander $G^{Ex}$ as follows. For each $(s,t) \in E^{Ex}$, we
add a flow between $s$ and $t$. Let this flow be routed along
possibly multiple paths. Let  $P_{stj}$ denote the $j^{th}$
path from $s$ to $t$ and $x_{stj}$ be
the amount of flow sent along this path. The length $\ell_{stj}$ of $P_{stj}$ is at most $W$
as the discussion in Lemma \ref{lemma:bmfp}. Let the overall solution,
denoted by $\{(\P_r,w_r)\}$, gives a feasible solution in
the following polytope with $x_{stj}$ as its variables:
\begin{align*}
&\sum_j x_{stj} = \pi_s P^{Ex}_{st}, ~~\forall (s, t) \in E^{Ex}\\
&\sum_{st \in E^{Ex}} \sum_{j:e\in \P_{stj}} x_{stj} \leq W Q_e, ~~ \forall e\in E\\
&x_{stj}\geq 0\ \ \ \ \forall s,t,j.
\end{align*}
Clearly, any feasible solution in this polytope, say
$\{(\P_r,w_r)\}$, will work for our lifting construction. Now, the
size of its support set is $|\{(\P_r,w_r)\}|$. If we consider the
extreme point of this polytope, the size of its support set is at
most $|E^{Ex}|+|E|=O(|E|)$ because the extreme point is an unique
solution of a sub-collection of linear constraints in this polytope.
Hence, if we choose such an extreme point $\{(\P_r,w_r)\}$ for our
lifting, the size of our lifted chain $|\widehat{E}|$ is at most
$O(W|E|)$ since each path is of length $O(W)$. Thus, we have
established that the size of the lifted Markov chain is at most
$O(W|E|)=O^*(|E|/\Phi(P))$.
\end{proof}

\vspace{-.1in}
\subsection{Useful Claims}
\vspace{-.1in}
We state and prove two useful claims which plays a key role in proving
Lemma \ref{lemma:mixing_expander}.
\begin{claim}\label{clm:a1}
Let $P_1,P_2$ be reversible Markov chains with
their stationary distributions $\pi_1,\pi_2$ respectively. If there exist
positive constants $\alpha,\beta,c,d$ such that
 $P_1\geq \alpha P_2$, $P_1\geq \beta
I$ and $c\pi_2 \leq \pi_1 \leq d\pi_2$, then
$$\lambda_{P_1}\geq \min \left(\frac{\alpha c}{d^2}\lambda_{P_2},2\beta \right).$$
\end{claim}
\begin{proof}
From the min-max characterization of the spectral gap (see, e.g., the page 176 in \cite{HJ86})
for the reversible Markov chain,
it follows that
\begin{align*}
\lambda_{P_1} & = ~\inf_{\psi:V \rightarrow \mathbb{R}} \lf(\frac{\sum_{i,j \in
V}(\psi(i)-\psi(j))^2(\pi_1)_i(P_1)_{ij}}{\sum_{i,j \in
V}(\psi(i)-\psi(j))^2(\pi_1)_i(\pi_1)_j}  \rf) \\
&\geq \lf(\frac{\alpha c}{d^2} \rf) \inf_{\psi:V \rightarrow \mathbb{R}}
\lf( \frac{\sum_{i,j \in
V}(\psi(i)-\psi(j))^2(\pi_2)_i(P_2)_{ij}}{\sum_{i,j \in
V}(\psi(i)-\psi(j))^2(\pi_2)_i(\pi_2)_j} \rf)\\
&=  \lf(\frac{\alpha c}{d^2}\rf) \lambda_{P_2}.
\end{align*}
The smallest eigenvalue of $P_1$ is greater than $2\beta-1$ because
$P_1\geq \beta I$. So, the distance between the smallest eigenvalue
and -1 is greater than $2\beta$. This completes the proof.
\end{proof}

\begin{claim}\label{clm:a2}
Let $P_1,P_2$ be Markov chains with their stationary distributions
$\pi_1,\pi_2$ respectively. Now, suppose $P_2$ is reversible. ($P_1$ is not necessarily reversible.) If there exist positive constants $\alpha,\beta,c,d$
such that $P_1\geq \alpha P_2$, $P_1\geq \beta I$ and
$c\pi_2 \leq \pi_1 \leq d\pi_2$, then
$$\lambda_{P_1P^*_1}\geq \min \left(\frac{\alpha \beta
c}{d^2}\lambda_{P_2},2\beta^2 \right).$$
\end{claim}
\begin{proof}
$P_1P^*_1$ is a reversible Markov chain which has $\pi_1$ as its
stationary distribution. Because $P^*_1\geq \beta I$, $P_1P^*_1\geq
\alpha P_2 P^*_1 \geq \alpha \beta P_2.$ Also, $P_1P^*_1\geq
\beta^2 I$. Now, the proof follows from
Claim \ref{clm:a1}.
\end{proof}

\section{Conclusion}\label{sec:conc}

Motivated by applications arising in emerging networks such as sensor networks,
peer-to-peer networks and surveillance network of unmanned vehicles, we consider
the question of designing fast linear iterative algorithms for computing the average of
numbers in a network. We presented a novel construction of such an algorithm by
designing the fastest mixing non-reversible Markov chain on any given graph. Our
Markov chain obtained through a new notion denoted by pseudo-lifting. We apply
our constructions to graphs with geometry, or graphs with doubling dimension. By
using their topological properties explicitly, we obtain fast and slim pseudo-lifted
Markov chains. The effectiveness (and optimality) of our constructions are explained
through various examples.  As a byproduct, our result provides the fastest
mixing Markov chain for any given graph which should be of interest in its own
right. Our result should naturally find their applications in the context of distributed
optimization, estimation and control.

We note that the pseudo-lifting presented here is based on
a two-level ``hierarchical star'' topology. This construction is
less robust to node failures. For example, failure of ``root''
node can increase the mixing time drastically. To address this,
one may alternatively use a ``hierarchical expander'' based
pseudo-lifting. That is, in place of the ``star'' topology in the
pseudo-lifting, utilize the ``expader'' topology. This will naturally
make the construction more robust without loss of performance.
Of course, this will complicate the mixing time analysis
drastically. This is where our method developed in the expander-based lifting will be readily useful.

\bibliographystyle{plain}
\bibliography{biblio}

\end{document}